\documentclass[10pt, aps,nofootinbib,superscriptaddress,notitlepage,tightenlines,prl,twocolumn]{revtex4-1}

 \usepackage{qip}
\usepackage{graphicx}
\usepackage[colorlinks=true,citecolor=blue,urlcolor=black]{hyperref}
\usepackage{amsfonts}
\usepackage{amsmath,amssymb}
\usepackage{dsfont}
\usepackage{euscript}
\usepackage{float}
\usepackage{amsthm}
\usepackage{bbm}
\usepackage{ragged2e}
\usepackage{color}
\usepackage{xcolor}

\usepackage{color,soul}
\usepackage{cleveref}
  
\crefname{table}{Protocol}{Protocols}
\usepackage{float}
\def\be{\begin{equation}}
\def\ee{\end{equation}}

\newcommand{\eqkd}{\eps_{\rm{qkd}}}
\newcommand{\epe}{\eps_{\rm{pe}}}
\newcommand{\epa}{\eps_{\rm{pa}}}

\newtheorem{theorem}{Theorem}

\newtheorem{lemma}[theorem]{Lemma}

\begin{document}

\title{Security analysis of quantum key distribution with small block length \\and its application to quantum space communications}

\author{Charles Ci-Wen Lim}\email{charles.lim@nus.edu.sg}
\affiliation{Department of Electrical \& Computer Engineering, National University of Singapore, 4 Engineering Drive 3, Singapore 117583}
\affiliation{Centre for Quantum Technologies, National University of Singapore, 3 Science Drive 2, Singapore 117543}

\author{Feihu Xu}
\affiliation{Hefei National Laboratory for Physical Sciences at Microscale and Department of Modern Physics, University of Science and Technology of China, Hefei 230026, China}
\affiliation{Shanghai Branch, CAS Center for Excellence in Quantum Information and Quantum Physics, University of Science and Technology of China, Shanghai 201315, China}

\author{Jian-Wei Pan}
\affiliation{Hefei National Laboratory for Physical Sciences at Microscale and Department of Modern Physics, University of Science and Technology of China, Hefei 230026, China}
\affiliation{Shanghai Branch, CAS Center for Excellence in Quantum Information and Quantum Physics, University of Science and Technology of China, Shanghai 201315, China}

\author{Artur Ekert}
\affiliation{Centre for Quantum Technologies, National University of Singapore, 3 Science Drive 2, Singapore 117543}
\affiliation{Mathematical Institute, University of Oxford, Oxford, Oxford OX2 6HD, United Kingdom}

\begin{abstract}
The security of real-world quantum key distribution (QKD) critically depends on the number of data points the system can collect in a finite time interval. To date, state-of-the-art finite-key security analyses require block lengths in the order of $10^4$ bits to obtain positive secret keys. This requirement, however, can be very difficult to achieve in practice, especially in the case of entanglement-based satellite QKD, where the overall channel loss can go up to 70 dB or more. Here, we provide an improved finite-key security analysis which reduces the block length requirement by 14\% to 17\% for standard channel and protocol settings. In practical terms, this reduction could save entanglement-based satellite QKD weeks of measurement time and resources, thereby bringing space-based QKD technology closer to reality. As an application, we use the improved analysis to show that the recently reported Micius QKD satellite is capable of generating positive secret keys with a $10^{-5}$ security level.
\end{abstract}

\maketitle
{\it{Introduction.---}}Quantum key distribution (QKD) is arguably the most mature quantum technology today with full-stack systems already being deployed in real-world networks~\cite{Scarani2009,Lo2014, Diamanti2016, Feihu2020RMP}. In the research community, the security of real-world QKD systems remains a central programme and much effort has been devoted to understanding the security of side-channels and systems with finite resources. Of particular importance is the non-asymptotic security of QKD---a theoretical programme that analyzes the security of finite-length keys and the security errors of parameter estimation, error correction, and privacy amplification~\cite{ren08}. 

In fact, the need for finite-key security was recognized long ago~\cite{Mayers2001,Inamori2007,Scarani2008} but it is only in the recent decade that security proof techniques became more accessible. On this note, we refer to those based on entropic uncertainty relations~\cite{TR11}, which are known to be pretty efficient in terms of the \emph{block length}~\cite{Tomamichel2012a,Tomamichel2017}. Here, block length is defined as the size of the sifted measurement data before parameter estimation; i.e., $m=n+k$, where $n$ is the \emph{raw key} length and $k$ is the number of bits used for parameter estimation. Previously, one would need $m > 10^5$ to get a positive secret key~\cite{Scarani2008}; now, positive keys can be obtained for $m \sim 10^4$ using proof techniques based on entropic uncertainty relations, giving an order of magnitude of improvement~\cite{Tomamichel2012a,Tomamichel2017}. 

For QKD systems deployed in standard optical networks, one can readily obtain their finite-key security since the measurement data size of these systems is typically large enough~\cite{Feihu2020RMP}. However, for QKD systems deployed in challenging environments such as space~\cite{Yin2017a,Yin2017b,Takenaka2017, Yin2020}, obtaining their finite-key security can be an issue. In particular, owning to high channel loss and environmental challenges, these systems may not gather enough measurement data in short time intervals. To overcome this issue, one can either spend more time to gather data (the simplest solution) or upgrade the QKD system to achieve a faster data collection rate. These solutions are straightforward for terrestrial-based QKD systems, but the same cannot be said for space-based QKD systems. More specifically, the latter can only communicate with ground stations over very precise time windows (essentially only over good passes) and upgrading in space may not be realistic with current technology~\cite{Bedington2017}. These limitations hence suggest that space-based QKD is a major scientific challenge. Indeed, this is especially the case for entanglement-based satellite QKD (EB S-QKD), which is known to be slower than the commonly adopted prepare-and-measure setup but offers the attractive possibility of untrusted node operation.

Very recently, it has been reported that the finite-key security of EB S-QKD has been demonstrated for the first time with an improved \emph{Micius} QKD satellite~\cite{Yin2020}. There, the authors reported a positive secret key with a security parameter of  $10^{-10}$ based on the security proof technique of Ref~\cite{Tomamichel2012a}. However, the analysis there did not consider the cost of parameter estimation and hence the random sampling bound was not properly applied. As a result, this led to overly-optimistic secret key rates. Using their experimental data ($m \sim 10^{3}$) and the methods from Refs.~\cite{Tomamichel2012a,Tomamichel2017}, we found that one would actually require a larger security parameter of  $10^{-5}$ in order to generate a positive key (see below).

In view of the above, it is thus timely to develop new proof techniques that would further sharpen the finite-key security of QKD systems with small block lengths, particularly those with $m \sim10^3$ to $10^4$. To that end, we present here an improved finite-key analysis which can further reduce the block length requirement by 14\% to 17\%. While the improvement is not dramatic (in the sense by several orders of magnitude), it could already save EB S-QKD many weeks of data collection and resources. Our main theoretical contribution is a new random sampling bound (without replacement), which improves pretty significantly the statistical accuracy of the parameter estimation procedure in the small block length regime. The inequality is presented in Lemma~\ref{lem:samplingbound} and may be of independent interest to other problems. Using this refined analysis, we show that the improved Micius QKD system~\cite{Yin2020} is capable of generating positive keys with $10^{-6}$ security; with the same level of security, no positive keys were found (numerically) with the analysis from Ref.~\cite{Tomamichel2017}.

{\it{Entanglement-based QKD.---}}In order to formally present our results, let us first introduce the BBM92 QKD protocol~\cite{BBM92} in consideration. We assume that the QKD users, Alice and Bob, are given a $2m$-partite quantum state $\rho_{AB}$, which is derived from a larger $(2m+1)$-partite quantum system $ABE$. Here, the quantum system $E$ is controlled by an adversary, Eve, who is computationally unbounded. We consider a protocol in which Alice's measurements are ideal  (i.e., incompatible qubit measurements), the basis choice is uniform, and the raw key is randomly sampled from both bases~\cite{Tomamichel2017}. This protocol is essentially the same as the BBM92 QKD protocol analyzed by Shor and Preskill~\cite{Shor2000} except that less than half of the sifted bits are used for parameter estimation (via random sampling): $n$ bits for the raw key and $k=m-n$ bits for parameter estimation. Importantly, in sampling from both bases, the errors will be uniformly distributed (and hence symmetrized) between the raw key pair and the random sample used for parameter estimation. This allows us to view the parameter estimation step as a standard random sampling problem, where the goal is to infer the error rate of the raw key pair using the error rate observed in the random sample (drawn without replacement).

The protocol is characterized by the block length, $m$, and three key distillation subroutines, i.e., parameter estimation, error correction, and privacy amplification, which we denote by $\mathsf{pe}$, $\mathsf{ec}$, and $\mathsf{pa}$, respectively. The parameter estimation subroutine is parameterized by two numbers: $\mathsf{pe}=\{k,\delta\}$, where $k\leq m$ is some positive integer and $\delta \in (0,1/2)$ is the tolerated error rate. The error correction and privacy amplification subroutines are parameterized by two positive integers and one positive integer, respectively: $\mathsf{ec}=\{t,r\}$ and  $\mathsf{pa}=\{\ell\}$. The QKD protocol is described in Protocol.~\ref{protocol:bbm92} (with only the details needed to illustrate the main technical results; for a more rigorous description, see Ref.~\cite{Tomamichel2017}).

\begin{table}[t!]
\hrule 
\smallskip
\textbf{{Settings: $m$, $\mathsf{pe}=\{k,\delta\}$, $\mathsf{ec}=\{t,r\}$ and  $\mathsf{pa}=\{\ell\}$ }}  
\smallskip
\hrule
\justify
\noindent\textbf{1.~Measurement.}~Alice and Bob agree on a random binary string $\Phi\in \{0,1\}^m$ over an authenticated public channel and measure their respective quantum signals using this string. They then agree on a random sample (drawn without replacement) of size $k$ from the entire measurement data set and store them into two pairs of strings: $(X,V)$ for Alice and $(Y,W)$ for Bob. Here, $X$ and $Y$ are random strings taking values in $\{0,1\}^n$; thus $V$ and $W$ take values in $\{0,1\}^k$. Note that $m=n+k$. \newline

\noindent\textbf{2.~Parameter estimation using random sampling.}~Alice publicly sends $V$ to Bob, who then computes the error rate between $V$ and $W$, i.e., $\overline{Z}_{\rm{pe}}:=|V \oplus W|/k$.  If the error rate exceeds the tolerated error rate $\delta$, they abort the protocol. Otherwise, they proceed to the next step. This decision is stored in a binary-valued flag $F_{\rm{pe}}\in\{\checkmark, \varnothing\}$, where $\checkmark$ means successful and $\varnothing$ means abort. \newline

\noindent\textbf{3.~Error correction.}~Alice sends Bob a syndrome $T$ of length $r$ which is computed from her raw key $X$. Then Bob generates an estimate of Alice's raw key, ${X'}$, from $Y$ and $T$. To verify that the correction is successful, Alice computes a hash $H(X)$ (of length $t$) of $X$ and sends it to Bob, who then compares it with his hash $H({{X'}})$. If the hash values are different (i.e., $H(X) \not=H({X'})$), they abort the protocol; this decision is stored in $F_{\rm{ec}}\in\{\checkmark, \varnothing\}$.  \newline

\noindent\textbf{4.~Privacy Amplification.}~Alice and Bob perform randomness extraction based on two-universal hashing to extract an identical secret key pair, $S_A$ and $S_B$, each of length $\ell$, from $X$ and $X'$, respectively. \newline

\hrule
 \caption{BBM92 QKD with random sampling as adapted from Ref.~\cite{Tomamichel2017}.
  \label{protocol:bbm92}}
\end{table}

{\it{Security analysis.---}}Following standard QKD security definition~\cite{MR09,PR14}, a QKD protocol is said to be $\eqkd$-\textit{secure} if it is $\eps_1$-\textit{correct} and $\eps_2$-\textit{secret}, where $0<\eps_1+\eps_2 \leq \eqkd$. The security parameter of the protocol is thus $\eqkd$. The correctness criterion captures the maximum probability that Alice's and Bob's output keys are different, and the secrecy criterion looks at how distinguishable Alice's output key is from the ideal output key (i.e., one that is uniformly distributed and independent of Eve's side-information $E$ and all classical messages sent over the public channel). Importantly, this choice of security definition guarantees that the protocol is \emph{universally composable}~\cite{BM04,MR09,PR14}: the pair of secret keys can be safely used in any cryptographic task, e.g., for encrypting messages, that requires a perfectly secret key.

With the definitions and relevant tools in place we are now ready to analyze the finite-key security of Protocol.~\ref{protocol:bbm92}. The question which we would like to address is the following: What is the maximum value of $\ell$, the output secret key length, for a desired choice of security parameter $\eqkd$ and protocol settings? To answer this question, we first have to identify the mathematical conditions for which the protocol is $\eqkd$-secure. These conditions (for the considered BBM92 QKD protocol) are given in Ref.~\cite[Theorem 2 and Theorem 3]{Tomamichel2017}, which we restate in a concise form below:

\begin{theorem}[Adapted from Ref.~\cite{Tomamichel2017}]\label{thm:sec} For the QKD protocol described above with fixed settings, it is $\eqkd$-secure if there exists $\nu \in (0,1/2-\delta]$ satisfying
\be \label{sec_conds}
  2^{-t} + 2\epe(\nu) + \epa(\nu)\leq\eqkd,  
\ee where the error functions due to privacy amplification and parameter estimation are defined as
\begin{eqnarray*}
 \epa(\nu)&:=&\frac{1}{2}\sqrt{2^{-n[1-h_2(\delta+\nu)]+r+t+\ell}},\\
  \epe(\nu)&:=&\exp{\left(-\frac{nk^2\nu^2}{m(k+1)}\right)},
\end{eqnarray*}respectively, and with $h_2(x):=-x\log x - (1-x)\log (1-x)$, the binary entropy function. 
\end{theorem}

To illustrate the physical meaning (and importance) of $\nu$, we first recall the parameter estimation step. There, Alice and Bob have to check if $\overline{Z}_{\rm{pe}}:=|V \oplus W|/k  \leq \delta$. If it is true, they proceed to the next step; otherwise they abort the protocol. The whole point of this exercise is to ensure that $\overline{Z}_{\rm{key}}:=|X \oplus Y|/n$ is very close to $\delta$ with very high probability. More formally, what we want is an exponential tail bound on the following event: $\overline{Z}_{\rm{pe}} \leq \delta \cap \overline{Z}_{\rm{key}} \geq \delta + \nu$. This describes the bad event that the parameter estimation step is successful (i.e., $F_{\rm{pe}}=\checkmark$) and that the error rate of the raw key pair exceeds $\delta$ by some constant positive term $\nu$. In Refs.~\cite{Tomamichel2012a, Tomamichel2017}, it has been shown that the probability of this event, $\mathbb{P}_{\rm{pe}}:=\Pr\left[ \overline{Z}_{\rm{pe}} \leq \delta \cap \overline{Z}_{\rm{key}} \geq \delta + \nu\right]$, is upper bounded by
\be \label{oldpe}
\mathbb{P}_{\rm{pe}}\leq  \exp{\left(-\frac{2nk^2\nu^2}{m(k+1)}\right)}= \epe(\nu)^2,
\ee 
which is based on Serfling's seminal work in probability inequalities for random sampling without replacement~\cite[Corollary 1.1]{serfling1974}.

As one can see from the above, the deviation term $\nu$ plays an important role in maximizing $\ell$: a smaller deviation term $\nu$ means a larger secret key. However, $\nu$ cannot be arbitrarily small---this would make $\epe(\nu)$ very large (which in turn is constrained by the security parameter, $\eqkd$). Evidently, there is a delicate interplay between these error functions and it would pay off very well here if we can further sharpen Eq.~\eqref{oldpe}. In the following, we present a new random sampling probability bound for parameter estimation, which may be of independent interest to other problems as well.

\begin{lemma}[New sampling bound]\label{lem:samplingbound}~Let $\overline{Z}_{\rm{key}}$ and $\overline{Z}_{\rm{pe}}$ be defined as above. Then for any $\nu> \xi >0$ such that $m(\delta+\xi) \in \mathbb{Z}^+$ and $n^2 (\nu-\xi)^2>1$, 
\be
\mathbb{P}_{\rm{pe}}  \leq \exp{\left( -\frac{2mk \xi^2}{n+1} \right)} + \exp{\left( -2\Gamma_{m(\delta+\xi)}((n \nu')^2 -1) \right) },
\ee where $\nu':=\nu-\xi$ and
\[
\Gamma_{m(\delta+\xi)  }:= \frac{1}{ m(\delta+\xi)  +1}+\frac{1}{m- m(\delta+\xi) +1}.
\] 
\end{lemma} 
\begin{proof}
To start with, let $Z_{\rm{blk}}=\sum_{i=1}^m Z_i$ be the total number of errors in the initial block of length $m$ and $\overline{Z}_{\rm{blk}}=Z_{\rm{blk}}/m$ be the average error. Note that in QKD this information about the total error is not known to the users; we are modeling it here only for technical reasons. With this, we can write $\Pr\left[ \overline{Z}_{\rm{pe}} \leq \delta \cap \overline{Z}_{\rm{key}} \geq \delta + \nu\right]$ as a sum of $\Pr\left[ \overline{Z}_{\rm{pe}} \leq \delta \cap \overline{Z}_{\rm{key}} \geq \delta + \nu \cap\overline{Z}_{\rm{blk}}\geq  \overline{Z}_{\rm{pe}}+\xi \right]$ and $\Pr\left[ \overline{Z}_{\rm{pe}} \leq \delta \cap \overline{Z}_{\rm{key}} \geq \delta + \nu \cap\overline{Z}_{\rm{blk}} <  \overline{Z}_{\rm{pe}}+\xi \right]$ for some $\xi>0$. The first term can be upper bounded using
\begin{multline*}
\Pr\left[ \overline{Z}_{\rm{pe}} \leq \delta \cap \overline{Z}_{\rm{key}} \geq \delta + \nu \cap\overline{Z}_{\rm{blk}}\geq  \overline{Z}_{\rm{pe}}+\xi \right]\\ \leq
 \Pr\left[\overline{Z}_{\rm{blk}}\geq  \overline{Z}_{\rm{pe}}+\xi \right] \leq \exp{\left( -\frac{2mk \xi^2}{n+1} \right)},
\end{multline*} where the second inequality is implied from Ref~\cite[Corollary 1.1]{serfling1974} (lower tail instead of upper tail). 

For the second term, we have that $\Pr[ \overline{Z}_{\rm{pe}} \leq \delta \cap \overline{Z}_{\rm{key}} \geq \delta + \nu \cap\overline{Z}_{\rm{blk}} <  \overline{Z}_{\rm{pe}}+\xi ] \leq \Pr[ \overline{Z}_{\rm{key}} \geq \delta + \nu \cap\overline{Z}_{\rm{blk}} <  \delta+\xi ]$ as the latter event is implied from the former event, and $\Pr[ \overline{Z}_{\rm{key}} \geq \delta + \nu \cap\overline{Z}_{\rm{blk}} <  \delta+\xi ] \leq \Pr[ \overline{Z}_{\rm{key}} \geq \delta + \nu \cap {Z}_{\rm{blk}} <  \lceil m(\delta+\xi) \rceil ]=\Pr[ \overline{Z}_{\rm{key}} \geq \delta + \nu \cap {Z}_{\rm{blk}} <  m_{\rm{err}} ]$, where $m_{\rm{err}}:=\lceil m(\delta+\xi) \rceil $. Generally, we may assume $m_{\rm{err}}= m(\delta+\xi)$, i.e., $\delta+\xi$ is chosen such that $m(\delta+\xi)$ is a positive integer. 

The above can be further written as $\Pr[ \overline{Z}_{\rm{key}} \geq \delta + \nu \cap {Z}_{\rm{blk}} <  m_{\rm{err}} ]
\leq \sum_{t=0}^{m_{\rm{err}}} \Pr[ \overline{Z}_{\rm{key}} \geq \delta + \nu \cap Z_{\rm{blk}}= t ]= \sum_{t=0}^{m_{\rm{err}}} \Pr[ Z_{\rm{blk}} =t ] \Pr[ \overline{Z}_{\rm{key}}  \geq \delta + \nu | Z_{\rm{blk}}  = t ]$. An upper bound can be obtained by noting that
$\sum_{t=0}^{m_{\rm{err}}} \Pr[ Z_{\rm{blk}} =t ] \Pr[ \overline{Z}_{\rm{key}}  \geq \delta + \nu | Z_{\rm{blk}}  = t ]
 \leq  \Pr\left[ \overline{Z}_{\rm{key}}\geq \delta + \nu |{Z}_{\rm{blk}} =m_{\rm{err}} \right]$. Notice that the last term is given by the upper tail probability of the hypergeometric distribution given the total number of errors ${Z}_{\rm{blk}} =m_{\rm{err}}$ is fixed, i.e., $\Pr[ \overline{Z}_{\rm{key}}\geq \delta + \nu | {Z}_{\rm{blk}} =m_{\rm{err}}]= \Pr[ \overline{Z}_{\rm{key}} \geq m_{\rm{err}}/m + (\nu-\xi) |{Z}_{\rm{blk}} =m_{\rm{err}} ]$. In the ideal case, the exact value of the upper tail probability can be computed (since the total number of errors is tightly bounded now), but this may introduce numerical inaccuracies due to the discrete nature of the distribution~\cite{Trong1993,Berkopec2007}. To that end, we use instead a probability inequality by Hush and Scovel~\cite{HS2005}, which provides pretty tight exponential bounds for hypergeometric distribution (at least in the settings for which we are interested in). This inequality reads
\begin{multline*}
\Pr[ \overline{Z}_{\rm{key}} \geq m_{\rm{err}}/m + (\nu-\xi) |{Z}_{\rm{blk}} =m_{\rm{err}}] \\\leq \exp{\left( -2\alpha_{m_{\rm{err}}}(n^2 (\nu-\xi)^2 -1) \right)},  
\end{multline*}
where \[
\alpha_{m_{\rm{err}}}:= \max \left\{ \frac{1}{n+1}+\frac{1}{k+1}, \frac{1}{m_{\rm{err}}+1}+\frac{1}{m-m_{\rm{err}}+1} \right\}.
\] Note that it is important for $\nu>\xi$ since the deviation term must be positive; more precisely, the sampling should satisfy $n^2 (\nu-\xi)^2>1$. We can in fact further simplify the bound by picking the right-side-hand term in $\alpha_{m_{\rm{err}}}$, i.e., $\exp{( -2\alpha_{m_{\rm{err}}}(n^2 (\nu-\xi)^2 -1) )} \leq \exp{( -2\Gamma_{m_{\rm{err}}}(n^2 (\nu-\xi)^2 -1))}$, where $ \Gamma_{m_{\rm{err}}}:= 1/(m_{\rm{err}} +1)+ 1/(m-m_{\rm{err}} +1)$. It can be verified that $\Gamma_{m_{\rm{err}}}$ decreases with increasing $m_{\rm{err}}$ and hence 
\[ \Pr[ \overline{Z}_{\rm{key}}\geq \delta + \nu | {Z}_{\rm{blk}} =m_{\rm{err}}] \leq e^{-2\Gamma_{m_{\rm{err}}}((n\nu')^2-1)}.\]
where $\nu':=\nu-\xi$ is a positive constant.
\end{proof}

\begin{figure}[t!]
\includegraphics[width=8.64cm]{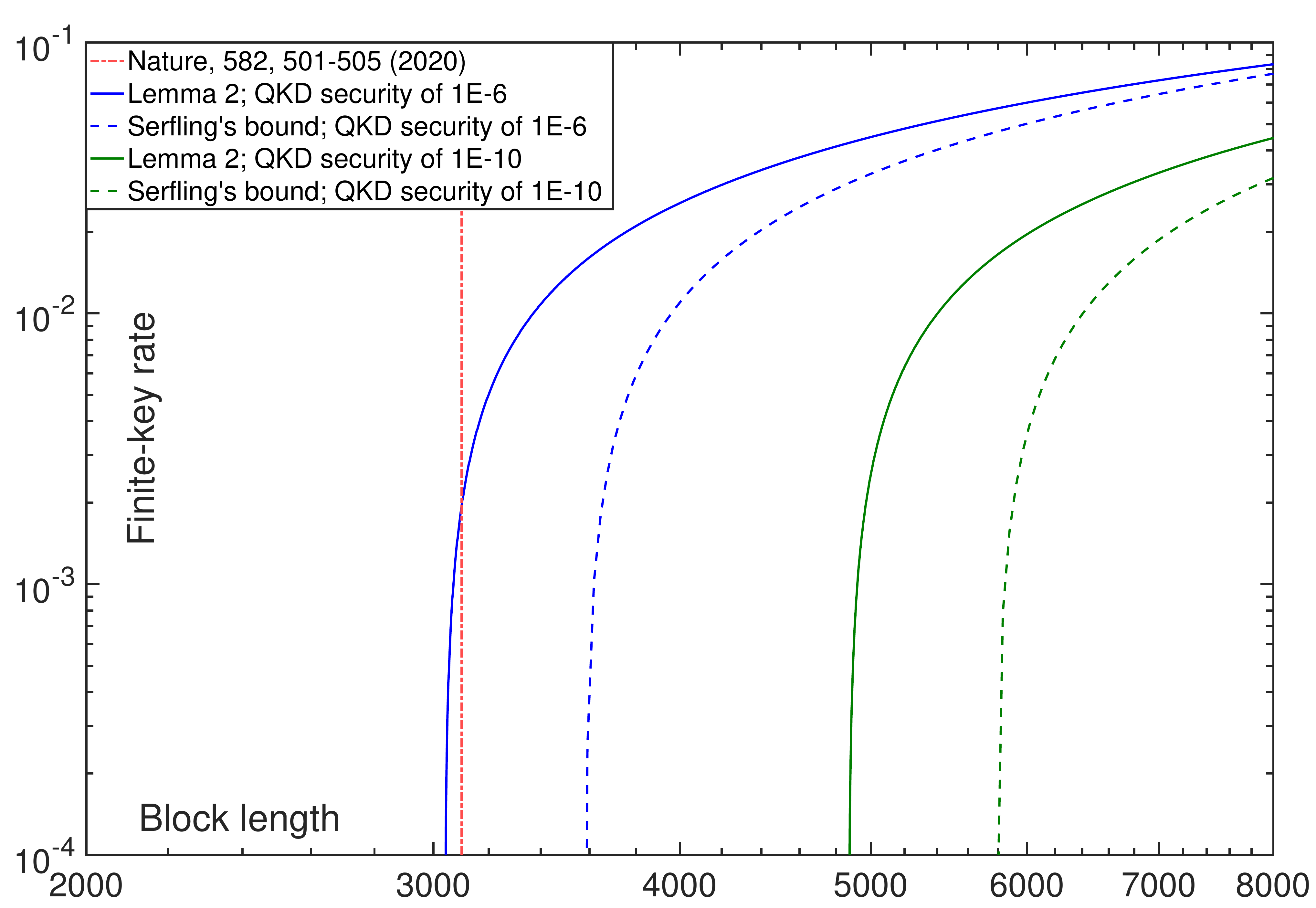}\caption{Numerically optimized finite-key rate $\ell/m$ vs block length $m$ for $s=6,10$. In this simulation, the tolerated error rate is set to $\delta=4.51\%$ Ref.~\cite{Yin2020}. The (red) vertical line represents the block length ($m=3100$ bits) obtained in the Micius experiment, which gives a finite-key rate of $1.962\times 10^{-3}$  based on $\eqkd=10^{-6}$. This suggests that about 6 secret bits can be generated in a single run of the protocol. The numerically optimized parameters are $ \vec{x}_{\rm{opt}} = \{1.962\times 10^{-3},0.5, 0.1141,0.0693\}$.} \label{fig1}
\end{figure}

Using Lemma.~\ref{lem:samplingbound}, the error function for parameter estimation now reads
\begin{multline*}
\epe'(\nu,\xi)^2 \\= \exp{\left( -\frac{2mk \xi^2}{n+1} \right)} + \exp{\left( -2\Gamma_{m(\delta+\xi)}((n \nu')^2 -1) \right) }. 
\end{multline*} To maximize the secret key length, $\ell$, we employ the following program parameterized by a bounded set of four-dimensional real vector, $\vec{x}=(\alpha,\beta,\nu,\xi)$. The block length $m$, tolerated error rate $\delta$, correctness error $2^{-t} = 10^{-(s+2)}$, and security parameter $\eqkd=10^{-s}$ are fixed.
\begin{eqnarray*}
& \underset{ {\vec{x}\in \mathbb{R}^4}}{\max}&  \quad \ell=
\lfloor \alpha m \rfloor \\
&\rm{s.t.}& \quad 2^{-t} + 2\epe'(\nu,\xi) + \epa(\nu)\leq\eqkd,\\
&&\quad \alpha \in [0,1], \beta \in (0,1/2],\\ 
&& \quad 0 < \xi < \nu <1/2-\delta,\\
\end{eqnarray*}
where $k=\lfloor \beta m \rfloor$ is the number of bits allocated to parameter estimation and $r=1.19 h_2(\delta)$ is the expected error correction leakage. Here, $\alpha$ is the secret key rate, which is defined as the number of secret bits generated divided by the block length, i.e., $\ell/m$. For comparison, we also run the same optimization for Theorem.~\ref{thm:sec}, which uses Eq.~\eqref{oldpe} for parameter estimation. From Figure~\ref{fig1}, we observe that the new analysis outperforms the existing analysis for both security levels, $s=6,10$: in particular, the minimum block length $m$ has been reduced by about $14\%$ to $17\%$. However, this improvement is not enough to obtain finite-key security for $\eqkd=10^{-10}$ as reported by Ref.~\cite{Yin2020}; for that, we would need $m > 4800$ (for the new analysis) and $m > 5800$ (for the one based on Eq.~\eqref{oldpe}). We note that one can still obtain finite-key security using Eq.~\eqref{oldpe}; however the smallest security error that we can find for a positive key is $\eqkd=10^{-5}$. 

{\it{Discussion and conclusion.---}}In this application study, we find that the improved Micius Satellite (arguably the most advanced QKD technology in space) would have to revise their security parameter from $\eqkd=10^{-10}$ to $10^{-6}$. This means that the overall security error is four orders of magnitude bigger than reported. This revision immediately raises the following question: Is $\eqkd=10^{-6}$ good enough for space-based QKD? Unfortunately, the answer to this question is not trivial and may entail discussions beyond the scope of this work. For the start, we emphasize that the security statement $\eqkd \leq 10^{-6}$ does not lend any concrete description to how the QKD system could have gone wrong. It is simply an upper bound on the probability that a bad event happens~\cite{PR14}; note that a QKD system which produces non-identical secret keys with probability $10^{-6}$ has the same security as one that gives the entire secret key to Eve with probability $10^{-6}$. We note that one could also use the adversary's guessing probability as a metric to help determine a \emph{good} value of $\eqkd$, e.g., see Ref.~\cite{Wang2020}.

Nevertheless, we argue that one can at least refer to \emph{key streaming} for some guidance. In this setting, we consider the use of a QKD system to generate a continuous stream of secret bits, where some portion of the earlier bits are used for authentication in subsequent rounds~\cite{MR09}. By composability, the security of the key stream is then $\eps_{\rm{stream}} \leq  v\eqkd$, where $v$ is the number of times the users intend to operate the QKD system. With this, we now have a clear connection to how many times the users can operate the QKD system before it needs to be rebooted with a fresh secret key. To appreciate this connection, suppose $\eps_{\rm{stream}}$ has to respect some recommended security level, say $10^{-5}$, which is specified by some national standards organization. Then, it is clear that for some fixed $\eqkd$ (determined by the system parameters), the number of times one can operate the QKD system, $v$, is limited. Applying this example to the improved Micius QKD satellite, we thus have $v \leq 10$, i.e., not more than 10 rounds of QKD operation. This suggests that the improved Micius QKD satellite may have to be rebooted (with a fresh secret key for authentication) after every 10 QKD cycles, or equivalently, the generation of 60 secret bits.

In conclusion, we have presented an improved finite-key analysis for BBM92 QKD protocol with small block length. While  the refined analysis gives pretty good improvements over existing methods, the problem of finite-key security with small block length is still a pressing one. This problem, as we have highlighted above, is especially relevant to long-range QKD systems such as satellite-based QKD, which will most likely form the backbone of the future quantum internet. 

{\it{Acknowledgements.--}}~We thank Valerio Scarani, Marco Tomamichel, Ernest Tan, Ignatius William Primaatmaja, Wang Chao, Yuan Cao, Juan Yin, Yu-Zhe Zhang, and Alexander Ling for their comments and suggestions. C. C.-W. Lim. is supported by the National Research Foundation (NRF) Singapore, under its NRF Fellowship programme (NRFF11-2019-0001) and Quantum Engineering Programme 1.0 (QEP-P2). F. Xu and J.-W. Pan acknowledge support by the National Key Research and Development (R\&D) Plan of China (2018YFB0504300), the National Natural Science Foundation of China, the Anhui Initiative in Quantum Information Technologies, the Shanghai Municipal Science and Technology Major Project (2019SHZDZX01).

\bibliography{references}

\end{document}